\documentclass[submission,copyright,creativecommons]{eptcs}
\providecommand{\event}{Non-Classical Logics 2024 (NCL'24)}

\usepackage{iftex}
\usepackage{amsmath, amssymb}
\usepackage{amsthm}
\usepackage{color}
\usepackage{url}
\usepackage[retainorgcmds]{IEEEtrantools}

\theoremstyle{definition}
\newtheorem{theorem}{Theorem} 
\newtheorem{definition}{Definition}
\newtheorem{proposition}[definition]{Proposition}
\newtheorem{lemma}[definition]{Lemma}
\newtheorem{corollary}[definition]{Corollary}
\newtheorem{remark}[definition]{Remark}

\newcommand{\Not}{{{\sim}}}

\newcommand{\Equiv}{{{\leftrightarrow}}}

\allowdisplaybreaks

\ifpdf
  \usepackage{underscore}         % Only needed if you use pdflatex.
  \usepackage[T1]{fontenc}        % Recommended with pdflatex
\else
  \usepackage{breakurl}           % Not needed if you use pdflatex only.
\fi

\title{Kamide is in America, \\Moisil and Leitgeb are in Australia\thanks{The research by Satoru Niki has received funding from the European Research Council (ERC) under the European Union's Horizon 2020 research and innovation programme, grant agreement ERC-2020-ADG, 101018280, ConLog. The research by Hitoshi Omori was supported by a Sofja Kovalevskaja Award of the Alexander von Humboldt-Foundation, funded by the German Ministry for Education and Research.}}
\author{Satoru Niki
\institute{Department of Philosophy I\\
Ruhr University Bochum\\
Bochum, Germany}
\email{Satoru.Niki@rub.de}
\and
Hitoshi Omori
\institute{Graduate School of Information Sciences\\
Tohoku University\\
Sendai, Japan}
\email{hitoshiomori@gmail.com}
}
\def\titlerunning{Kamide is in America, Moisil and Leitgeb are in Australia}
\def\authorrunning{S. Niki \& H. Omori}
\begin{document}
\maketitle

\begin{abstract}
It is not uncommon for a logic to be invented multiple times, hinting at its robustness. This trend is followed also by the expansion {\bf BD+} of Belnap-Dunn logic by Boolean negation. Ending up in the same logic, however, does not mean that the semantic interpretations are always the same as well. In particular, different interpretations can bring us to different logics, once the basic setting is moved from a classical one to an intuitionistic one. For {\bf BD+}, two such paths seem to have been taken; one ({\bf BDi}) by N. Kamide along the so-called American plan, and another ({\bf HYPE}) by G. Moisil and H. Leitgeb along the so-called Australian plan. The aim of this paper is to better understand this divergence. This task is approached mainly by (i) formulating a semantics for first-order {\bf BD+} that provides an Australian view of the system; (ii) showing connections of the less explored (first-order) {\bf BDi} with neighbouring systems, including an intermediate logic and variants of Nelson's logics.
\end{abstract}

\section{Introduction}

Since the birth of modern logic, with an enormous help from mathematical tools, we have seen many important and interesting formal theories being developed. Among the vast number of formal theories in the literature, those that are based on classical logic and intuitionistic logic have been particularly successful and explored in great depth. 

Soon after the initial developments of intuitionistic logic and theories based on it, there were a number of attempts in comparing the theories based on classical logic and theories based on intuitionistic logic. These comparisons, in many cases, are highly non-trivial, and sometimes even surprising. For example, take one of the most famous modal logic {\bf S5}. Then, it turns out that there are \emph{uncountably} many systems of intuitionistic version of {\bf S5} that will all collapse into classical {\bf S5} once one of the familiar formulas (e.g. the law of excluded middle, elimination of double negation, or Peirce’s law, and others) are added to the intuitionistic versions (cf. \cite[Corollary 2.4]{Ono1977}). Corresponding intuitionistic versions, therefore, of various formal theories may come along with a lot of surprising results, and also seem to bring us some new insights towards a deeper understanding of theories based on classical logic. 

In the present article, we will focus on the system {\bf BDi} developed by Norihiro Kamide in \cite{kamide2021modal}. In brief, {\bf BDi} is an intuitionistic version of the system {\bf BD+} which can be seen in at least two different ways: (i) as an expansion of classical logic by de Morgan negation, or (ii) as an expansion of {\bf FDE} (or Belnap-Dunn logic), expanded by Boolean negation. As we shall point out later in some more details, various systems that are definitionally equivalent to the system {\bf BD+} have been developed independently by various authors, and that seems to partly confirm the naturalness and importance of the system {\bf BD+}. Therefore, Kamide's attempt of investigating the intuitionistic version of {\bf BD+} seems to be of importance. 

Furthermore, as the title may already make some of the readers guess, there are interesting ways to connect Kamide's {\bf BDi} to yet another expansion of intuitionistic logic that has been known and studied by a few authors. Very roughly put, what is nowadays best known as {\bf HYPE}, (re)introduced by Hannes Leitgeb in \cite{Leitgeb2019hype}, though already introduced by Grigore Constantin Moisil in 1942, can be seen as another system that can be seen as an intuitionistic counterpart of {\bf BD+} (see \cite{drobyshevich2022moisil} for a detailed view of Moisil's work).
%\footnote{For a detailed presentation of Moisil's work, see \cite{drobyshevich2022moisil}.} 
Somewhat more precisely, Kamide's {\bf BDi} can be viewed as an intuitionistic counterpart of {\bf BD+} in light of the American plan for negation in {\bf FDE}, while the system explored by Moisil and Leitgeb can be viewed as an intuitionistic counterpart of {\bf BD+} in light of the Australian plan for negation in {\bf FDE}.

Against these backgrounds, the aim of this article is twofold. First, we will clarify the relations of systems {\bf BD+}, {\bf HYPE}, and {\bf BDi}. To this end, we will present another semantics for {\bf BD+} that offers a systematic view on the systems related to {\bf BD+}. Second, we will explore a few extensions and variations of {\bf BDi}, and in particular, establish some basic results for the extension of {\bf BDi} obtained by adding the \emph{ex contradictione quodlibet}. Most of our results are obtained for the language with first-order quantifiers.

\section{Semantics and proof system for {\bf BD+}}

The predicate language $\mathcal{L}_{Q}$ consists of connectives $\{\bot,\Not,\land,\lor,\to\}$, quantifiers $\{\forall,\exists\}$, countable sets of constants $\mathsf{Con}=\{c_{1},c_{2},\ldots\}$, variables $\mathsf{Var}=\{v_{1},v_{2},\ldots\}$ and $n$-ary predicates $\mathsf{Pred}=\{P^{n}_{1},P^{n}_{2},\ldots: n\in\mathbb{N}\}$. A \emph{term} is either a constant or a variable. The set of formulas in $\mathcal{L}_{Q}$ will be denoted by $\mathsf{Form}_{Q}$.

\vspace{-2mm}

\subsection{Preliminaries}
Let us recall the semantics in \cite[Definition 18]{Kamide2017extended}, for which we take $\bot$ and not $\neg$ as primitive here.

\begin{definition}
A {\bf QBD+}-Dunn-model for the language $\mathcal{L}_{Q}$ is a pair $\langle D, V \rangle$ where $D\supseteq \textsf{Con}$ is a non-empty set and we 
assign both the {\em extension} $V^{+}(P^n) \subseteq D^{n}$ and the {\em anti-extension} $ V^{-}(P^n) \subseteq D^{n}$ to each $n$-ary predicate symbol $P^n$. Valuations $V$ are then extended to interpretations $I$ for all the sentences of $\mathcal{L}_{Q}$ ($\textsf{Sent}_{Q}$) expanded by $D$
inductively as follows: as for the atomic {\em sentences},
\begin{itemize}
\setlength{\parskip}{0cm}
\setlength{\itemsep}{0cm}
\item $1{\in}I(P^n(t_{1},...,t_{n})) \textrm{ iff } \langle t_{1},\dots,t_{n} \rangle{\in}V^{+}(P^n)$,
\item $0{\in}I(P^n(t_{1},...,t_{n})) \textrm{ iff } \langle t_{1},\dots,t_{n} \rangle {\in}V^{-}(P^n)$.
\end{itemize}
The rest of the clauses are as follows:
\[
\begin{array}{llllll}
1 \not\in I(\bot), & & &0 \in I(\bot), &  & \\
1 \in I(\Not A) &\textrm{iff} & 0 \in I(A), &0 \in I(\Not A) &\textrm{iff} & 1 \in I(A),\\
1 \in I(A \land B) &\textrm{iff} & 1 \in I(A) \textrm{ and } 1 \in I(B), &0 \in I(A \land B) &\textrm{iff} & 0 \in I(A) \textrm{ or } 0 \in I(B),\\
1 \in I(A \lor B) &\textrm{iff} & 1 \in I(A) \textrm{ or } 1 \in I(B), &0 \in I(A \lor B) &\textrm{iff} & 0 \in I(A) \textrm{ and } 0 \in I(B),\\
1 \in I(A {\to} B) &\textrm{iff} & 1 \not\in I(A) \textrm{ or } 1 \in I(B), &0 \in I(A {\to} B) &\textrm{iff} & 0 \not\in I(A) \textrm{ and } 0 \in I(B),\\
1 \in I(\forall x A) &\textrm{iff} & 1 \in I(A(d)), \textrm{for all $d \in D$}, &0 \in I(\forall x A) &\textrm{iff} & 0 \in I(A(d)), \textrm{for some $d \in D$},\\
1 \in I(\exists x A) &\textrm{iff} & 1 \in I(A(d)), \textrm{for some $d \in D$}, &0 \in I(\exists x A) &\textrm{iff} & 0 \in I(A(d)), \textrm{for all $d \in D$}.
\end{array}
\]
Finally, let $\Gamma \cup \{ A \}$ be any set of sentences. Then, $A$ is a {\em {\bf BD+}-semantic consequence} from $\Gamma$ $(\Gamma \models A)$ iff for all {\bf QBD+}-Dunn-models $\langle D,V \rangle$, $1 \in I(A)$ if $1 \in I(B)$ for all $B \in \Gamma$. 
\end{definition}

\begin{remark}
Note that the unary operation $\neg A$ defined as $A{\to} \bot$ is Boolean Negation in the sense that:% the following holds:
\begin{itemize}
\setlength{\parskip}{0cm}
\setlength{\itemsep}{0cm}
\item $1 \in I(\neg A)$  iff $1 \not\in I(A)$, and $0 \in I(\neg A)$  iff $0 \not\in I(A)$.
%\item $0 \in I(\neg A)$  iff $0 \not\in I(A)$.
\end{itemize}
For a discussion on the notion of classical negation in {\bf FDE} and their extensions, see \cite{DeOmori15}.\footnote{For those who are ready to accept \emph{non-deterministic} classical negation, see also \cite{SzmucOmori2022}.}

Moreover, note that we have the following equivalences.
\begin{itemize}
\setlength{\parskip}{0cm}
\setlength{\itemsep}{0cm}
\item $1 \in I(\Not (\Not B {\to}\Not A))$  iff %$0 \in I(\Not B {\to}\Not A)$ iff $0 \not\in I(\Not B)$ and $0\in I(\Not A)$ iff 
$1\in I(A)$ and $1 \not\in I(B)$, and $0 \in I(\Not (\Not B {\to}\Not A))$  iff %$1 \in I(\Not B {\to}\Not A)$ iff $1 \not\in I(\Not B)$ or $1\in I(\Not A)$ iff 
$0 \in I(A)$ or $0 \not\in I(B)$.
\end{itemize}
Therefore, the connective $\leftarrow$ of the system {\bf SPL} introduced by Kamide and Wansing in \cite{Kamide2010symmetric} is definable in {\bf BD+}. This implies that {\bf SPL} and {\bf BD+} are definitionally equivalent.
\end{remark}

\begin{remark}
As already observed in \cite[\S3.5]{DeOmori15}, there are a few systems in the literature that are definitionally equivalent to {\bf BD+}. Those include, the system {\bf PM4N} formulated in the language $\{ \neg , \land , \lor , \Box \}$ by Jean-Yves B\'eziau in \cite{JYB2011}, and the system {\bf FDEP} formulated in the language $\{ \Not , \to \}$ by Dmitry Zaitsev in \cite{Zaitsev2012}. We already added another system {\bf SPL} in the previous remark, and we may add another more recent rediscovery by Arnon Avron. More specifically, Avron, in \cite{Avron2020normal}, introduces the system {\bf SE4} in the context of exploring expansions of {\bf FDE} by a conditional that are self-extensional.
\end{remark}

We now turn to the proof system, again recalling the definition and completeness theorem from \cite{Kamide2017extended}.
\begin{definition}
Consider the following axioms and rules where $\neg A$ and $A\Equiv B$ abbreviate $A{\to} \bot$ and $(A{\to} B){\land} (B{\to} A)$ respectively:

%\noindent

\vspace{-6mm}
\begin{minipage}{.5\textwidth}
\begin{align*}
& A {\to} (B {\to} A) \tag{Ax1} \label{Ax1} \\
& (A {\to} (B {\to} C)){\to} ((A {\to} B)  {\to} (A {\to} C)) \tag{Ax2} \label{Ax2} \\
& ((A {\to} B){\to} A){\to} A \tag{Ax3} \label{Ax3} \\
& (A \land B) {\to} A \tag{Ax4} \label{Ax4} \\
& (A \land B) {\to} B \tag{Ax5} \label{Ax5}\\
& (C{\to} A) {\to} ((C{\to} B) {\to} (C{\to} (A {\land} B))) \tag{Ax6} \label{Ax6}\\
& A {\to} (A \lor B) \tag{Ax7} \label{Ax7} \\
& B {\to} (A \lor B) \tag{Ax8} \label{Ax8}\\
& (A {\to} C) {\to} ((B {\to} C) {\to} ((A {\lor} B) {\to} C)) \tag{Ax9} \label{Ax9} \\
& \bot {\to} A \tag{Ax10} \label{Ax10} \\
& \frac{\ A\quad A{\to} B \ }{B} \label{MP} \tag{MP}
\end{align*}
\end{minipage}
\ 
\begin{minipage}{.4\textwidth}
\begin{align*}
& A(t)\to\exists{x}A \tag{Ax11}  \label{Ax11}\\
& \forall{x}(A{\to} B)\to(\exists{y}A(y){\to} B) \tag{Ax12} \label{Ax12}\\
& \forall{x}(B{\to} A){\to}(B{\to}\forall{x}A) \tag{Ax13} \label{Ax13}\\
& \forall{x}A{\to} A(t) \tag{Ax14} \label{Ax14}\\
& A {\to} \Not \bot \tag{Ax15} \label{deMbot}\\
& \Not \Not A\Equiv A \tag{Ax16} \label{deMN}\\
& \Not (A\land B)\Equiv (\Not A\lor \Not B) \tag{Ax17} \label{deMC} \\
& \Not (A\lor B)\Equiv (\Not A\land \Not B) \tag{Ax18} \label{deMD}\\
& \Not (A{\to} B)\Equiv (\neg\Not A{\land} \Not B) \tag{Ax19} \label{deMI}\\ 
& \Not\forall{x}A{\leftrightarrow}\exists{x}\Not A \tag{Ax20} \label{deMU}\\
& \Not\exists{x}A{\leftrightarrow}\forall{x}\Not A \tag{Ax21} \label{deMP}\\
& \frac{\ A \ }{\forall{x}A} \label{Gen} \tag{Gen}
\end{align*}
\end{minipage}

\vspace{2mm}

\noindent We write $\Gamma\vdash A$ if there is a finite list $B_{1},\ldots,B_{n}\equiv A$ such that each $B_{i}$ is either an element of $\Gamma$, an instance of one of the axioms, or obtained from previous items in the list by \eqref{MP} or \eqref{Gen}.
\end{definition}

%Then, the following result is established in \cite[Theorem 25]{Kamide2017extended}.

\begin{theorem}\label{thm:QBD+Dunn}
For all $\Gamma\cup\{A\}\subseteq\mathsf{Sent}_{Q}$, $\Gamma\vdash A$ iff $\Gamma\models A$.
\end{theorem}

\subsection{Another semantics}
Before moving ahead, let us introduce another semantics for {\bf BD+}.\footnote{The propositional fragment is already introduced briefly in \cite{Omori2019sixteen}.}

\begin{definition}
A {\bf QBD+}-star-model for the language $\mathcal{L}_{Q}$ is a quadruple $\langle W, \ast, D, V \rangle$ where $W$ is a non-empty set (of states); $\ast$ is a function on $W$ with $w^{\ast\ast}=w$ for all $w\in W$; $D\supseteq \textsf{Con}$ is a non-empty set and we assign 
%$V(c) \in D$ to each constant $c$, assign 
the {\em extension} $V(w, P^n) \subseteq D^{n}$ to each $n$-ary predicate symbol $P^n$ and $w\in W$. Valuations $V$ are then extended to interpretations $I$ for all the state-sentence pairs of $\mathcal{L}$ expanded by $D$
%$\{ k_d : d \in D \}$ 
inductively as follows: as for the atomic {\em sentences}, 
\begin{itemize}
\item $I(w, P^n(t_{1},...,t_{n}))=1$ iff $\langle t_{1},\dots,t_{n} \rangle \in V(w, P^n)$.
%\item $I(w, P^n(t_{1},...,t_{n}))=1$ iff $\langle V(t_{1}),\dots,V(t_{n}) \rangle \in V(w, P^n)$.
\end{itemize}
The rest of the clauses are as follows:

\smallskip

\noindent 
\begin{minipage}{.47\textwidth}
\begin{itemize}
\setlength{\parskip}{0cm}
\setlength{\itemsep}{0cm}
\item $I(w, \bot)\neq 1$, 
\item $I(w, \Not A)=1$ iff $I(w^\ast, A)\neq 1$,
\item $I(w, A {\land} B){=}1$ iff $I(w, A){=}1$ and $I(w, B){=}1$,
\item $I(w, A {\lor} B){=}1$ iff $I(w, A){=}1$ or $I(w, B){=}1$,
\end{itemize}
\end{minipage}
\begin{minipage}{.52\textwidth}
\begin{itemize}
\setlength{\parskip}{0cm}
\setlength{\itemsep}{0cm}
\item $I(w, A {\to} B){=}1$ iff $I(w, A){\neq} 1$ or $I(w, B){=}1$,
\item $I(w, \forall x A){=}1$ iff $I(w, A(d)){=}1$, for all $d {\in} D$, 
\item $I(w, \exists x A){=}1$ iff $I(w, A(d)){=}1$, for some $d {\in} D$.
\end{itemize}
\end{minipage}

\smallskip
\noindent Finally, let $\Gamma \cup \{ A \}$ be any set of sentences. Then, $A$ is a {\em {\bf BD+}-star-semantic consequence} from $\Gamma$ $(\Gamma \models_\ast A)$ iff for all {\bf QBD+}-star-models $\langle W, \ast, D, V \rangle$, and for all $w\in W$, $I(w, A)=1$ if $I(w, B)=1$ for all $B \in \Gamma$. 
\end{definition}

Then, we obtain the following result.

\begin{proposition}
For all $\Gamma\cup\{A\}\subseteq\mathsf{Sent}_{Q}$, $\Gamma\vdash A$ iff $\Gamma\models_\ast A$.
\end{proposition}
\begin{proof}
For the soundness direction, we will only check the case for \eqref{deMI}. For all $A, B\in \mathsf{Sent}_{Q}$ and for all $w\in W$: $I(w, \Not (A{\to} B)){=}1$ iff $I(w^\ast, A{\to} B){\neq} 1$ iff $I(w^\ast, A){=}1$ and $I(w^\ast, B){\neq} 1$ iff $I(w, \neg \Not A){=}1$ and $I(w, \Not B){=} 1$ iff $I(w, \neg \Not A\land \Not B){=} 1$. Therefore, we obtain the desired result.

For the completeness direction, it suffices to show that $\Gamma\models_\ast A$ only if $\Gamma\models A$ by Theorem~\ref{thm:QBD+Dunn}. Suppose $\Gamma\not\models A$. Then, there is a {\bf QBD+}-Dunn-model $\langle D_0, V_0 \rangle$ such that $1 {\not\in} I_0(A)$ and $1 {\in} I_0(B)$ for all $B {\in} \Gamma$. Define a {\bf QBD+}-star-model $\langle W_1, \ast_1, D_1, V_1 \rangle$ as follows: $W_1{:=} \{ a, b \}$; $a^\ast{=}b, b^\ast{=}a$; $D_1{:=} D_0$; $V_1(a, P^n){:=} V_0^+(P^n)$, $V_1(b, P^n){:=} D^n\setminus V_0^-(P^n)$.
%\begin{itemize}
%\setlength{\parskip}{0cm}
%\setlength{\itemsep}{0cm}
%\item $W_1:= \{ a, b \}$;
%\item $a^\ast=b, b^\ast=a$;
%\item $D_1:= D_0$; 
%\item $V_1(a, P^n):= V_0^+(P^n)$, $V_1(b, P^n):= D^n\setminus V_0^-(P^n)$
%\end{itemize}
Then, we can show that the following holds for all sentences: 
\begin{itemize}
\setlength{\parskip}{0cm}
\setlength{\itemsep}{0cm}
\item $I_1(a, A)=1$ iff $1\in I_0(A)$ and $I_1(b, A)=1$ iff $0\not\in I_0(A)$
%\item $I_1(b, A)=1$ iff $0\not\in I_0(A)$
\end{itemize}
We can prove this by induction, but the details are straightforward and safely left to the readers. We are then ready to conclude that $\Gamma\not\models_\ast A$ since we have $I_1(a, A)\neq 1$ and $I_1(a, B)=1$ for all $B \in \Gamma$ in the {\bf QBD+}-star-model $\langle W_1, \ast_1, D_1, V_1 \rangle$. This completes the proof.
\end{proof}

\begin{remark}
Both for {\bf SPL} and {\bf SE4}, the status of the contraposition rule is highlighted, and this becomes even clearer once we have the star semantics. We may also add that our proof can be seen as an alternative proof to the result on the admissibility of contraposition rule in {\bf BD+} established by Kamide in \cite[Theorem 16]{Kamide2022} in which two sequent calculi are made use of. 

Moreover, the star semantics makes the relation between {\bf HYPE} and {\bf BD+} (and its definitionally equivalent systems) explicit. Indeed, by building on the semantics for {\bf HYPE} presented by Sergei Odintsov and Heinrich Wansing in \cite{odintsovWansing2021}, it is easy to see that {\bf BD+} is obtained by trivialising the partial order which is necessary to capture the constructive conditional.
\end{remark}

\section{{\bf N3}-style extension of {\bf BDi}}
In \cite{kamide2021modal}, Norihiro Kamide presented an intuitionistic version of the system {\bf BD+}. This variant {\bf BDi} can also be seen as a variant of the system {\bf N4} of Almukdad and Nelson \cite{almukdad1984}, obtained by changing the falsity condition for implication. It then is a natural question to study an extension of {\bf BDi} with the characteristic axiom for {\bf N3} \cite{nelson1949}, the explosive variant of {\bf N4}. We shall see that this extension, henceforth called {\bf BDi3}, validates the principle of \emph{potential omniscience} investigated by Ichiro Hasuo and Ryo Kashima \cite{hasuo2003kripke}, in contrast to the case for {\bf N3}. This motivates us to consider {\bf BDi3} as a predicate logic {\bf QBDi3}, since potential omniscience implies the \emph{double negation shift} (a.k.a. \emph{Kuroda's conjecture}) $\forall{x}\neg\neg A\to\neg\neg\forall{x}A$.

\subsection{Semantics}

\begin{definition}
A {\bf QBDi3}-model for the language $\mathcal{L}_{Q}$ is a quadruple $\langle W,\leq, D, V\rangle$, where $W$ is a non-empty set (of states); $\leq$ is a partial ordering on $W$; $D$ is a mapping that assigns to each $w\in W$ a set $D(w)\supseteq\mathsf{Con}$, with a proviso that $x\geq w$ implies $D(x)\supseteq D(w)$. As an additional condition, $(W,\leq)$ has to satisfy $\forall{w\in W}\exists{x\geq w}(\forall{y}(y\geq x\Rightarrow y=x))$, i.e. any state has a maximal successor.\\
\indent %Let $\mathbf{D}=\cup_{w\in W} D(w)$. 
$V$ assigns both the {\em extension} $V^{+}(w,P^n) \subseteq (D(w))^{n}$ and the {\em anti-extension} $V^{-}(w,P^n) \subseteq (D(w))^{n}$ to each $n$-ary predicate symbol $P^n$ and a state $w$, such that $V^{+}(w,P^n)\cap V^{-}(w,P^n)=\emptyset$. Moreover, $V^{+}$ and $V^{-}$ must be monotone: $\langle d_{1},\ldots d_{n}\rangle\in V^{*}(w,P^n)$ and $x\geq w$ implies  $\langle d_{1},\ldots d_{n}\rangle\in V^{*}(x,P^n)$ for $*\in\{+,-\}$.
Additionally, we assume $V$ to be \emph{potentially omniscient}, i.e. for all $w\in W$ and $\langle d_{1},\ldots,d_{n}\rangle\in (D(w))^{n}$: for all $x\geq w$ there exists $y\geq x$: $\langle d_{1},\ldots,d_{n}\rangle\in V^{+}(y, P^n)\cup V^{-}(y,P^n)$. $V$ is extended to the interpretation $I$ to state-sentence pairs (of $\mathsf{Sent}_{{\bf D}}$, i.e. $\mathcal{L}_{Q}$ extended with $\mathbf{D}:=\cup_{w\in W} D(w)$) by the following conditions:

\begin{itemize}
\setlength{\parskip}{0cm}
\setlength{\itemsep}{0cm}
%\item $I(w,P(d_{1},\ldots,d_{n}))=V(w,P(d_{1},\ldots,d_{n}))$,
\item $1\in I(w,P(d_{1},\ldots,d_{n}))$ iff $\langle d_{1},\ldots,d_{n}\rangle\in  V^{+}(x,P^n)$,
\item $0\in I(w,P(d_{1},\ldots,d_{n}))$ iff $\langle d_{1},\ldots,d_{n}\rangle\in  V^{-}(x,P^n)$,
\item $1\notin I(w, \bot)$ and $0\in I(w,\bot)$,
\item $1\in I(w, \Not A)$ iff $0\in I(w, A)$,
\item $0\in I(w, \Not A)$ iff $1\in I(w, A)$,
\item $1\in I(w, A\land B)$ iff $1\in I(w, A)$ and $1\in I(w, B)$,
\item $0\in I(w, A\land B)$ iff $0\in I(w, A)$ or $0\in I(w, B)$,
\item $1\in I(w, A\lor B)$ iff $1\in I(w, A)$ or $1\in I(w, B)$,
\item $0\in I(w, A\lor B)$ iff $0\in I(w, A)$ and $0\in I(w, B)$,
\item $1\in I(w, A{\to} B)$ iff for all $x\in W:$ ($w\leq x$ only if ($1\notin I(x, A)$ or $1\in I(x, B)$)),
\item $0\in I(w, A{\to} B)$ iff for all $x\in W:$ (($w\leq x$ only if $0\notin I(x, A)$) and $0\in I(w,B))$,
%\item $1\in I(w, \neg A)$ iff for all $x\in W:$ ($w\leq x$ only if $1\notin I(x, A)$),
%\item $0\in I(w, \neg B)$ iff for all $x\in W:$ ($w\leq x$ only if $0\notin I(x, A)$),
\item $1\in I(w, \forall{x}A)$ iff for all $x\in W:$ ($w\leq x$ only if $1\in I(x, A(d))$ for all $d\in D(x)$),
\item $0\in I(w, \forall{x}A)$ iff $0\in I(w, A(d))$ for some $d\in D(w)$,
\item $1\in I(w, \exists{x}A)$ iff $1\in I(w, A(d))$ for some $d\in D(w)$,
\item $0\in I(w, \exists{x}A)$ iff for all $x\in W:$ ($w\leq x$ only if  $0\in I(x, A(d))$ for all $d\in D(x)$).
\end{itemize}
Finally, the semantic consequence is defined as follows: $\Gamma \models_{i3} A$ iff for all {\bf QBDi3}-models $\langle W, \leq, D, V \rangle$, and for all $w\in W$: $1\in I(w, A)$ if $1\in I(w, B)$ for all $B\in \Gamma$.
\end{definition}

\begin{remark}
Let $\mathcal{L}_{int}$ be a language consisting of $\{\bot,\Not\bot,\land,\lor,\to,\forall,\exists\}$ and containing additional predicates $P',Q',$ etc. corresponding to $P,Q,$ etc. We include  $\Not\bot$ for the sake of convenience in the proof of completeness.
Then a model of intuitionistic logic \emph{plus} double negation shift, known as {\bf MH}, can be defined by restricting the language to $\mathcal{L}_{int}$, removing references to $V^{-}$, $\Not$-related clauses and $0$ in the interpretation and adding the clause that $1\in I(w,\Not\bot)$.
We shall use $\models_{mh}$ to denote the consequence.
\end{remark}

The following proposition can be established by induction on the complexity of $A$.

\begin{proposition}\label{prop.model}
In a {\bf QBDi3}-model, for all $A\in\mathsf{Sent}_{{\bf D}}$, if $1\in I(w,A)$ and $w\leq x$ then $1\in I(x,A)$.
\end{proposition}

%\begin{proof}
%By induction on the complexity of $A$.
%\end{proof}

\begin{proposition}\label{prop.cond}
In a {\bf QBDi3}-model, for all $w\in W$ the following statements hold.\\
\textnormal{(i)} For no $A(\vec{d})\in\mathsf{Sent}_{{\bf D}}$ s.t. $\vec{d}\in D(w)$, $1\in I(w,A(\vec{d}))$ and $0\in I(w,A(\vec{d}))$,\\
\textnormal{(ii)} For all $A(\vec{d})\in\mathsf{Sent}_{{\bf D}}$ s.t. $\vec{d}\in D(w)$, for all $x\geq w$ there exists $y\geq x:$ ($1\in I(y,A(\vec{d}))$ or $0\in I(y,A(\vec{d}))$). 
\end{proposition}

\begin{proof}
By simultaneous induction on the complexity of $A$. Here we shall look at the case for $\to$ 
and $\forall$.
\\
\textbf{For implication:}
(i) Suppose $1\in I(w,B\to C)$ and $0\in I(w, B\to C)$. By IH, for all $x\geq w$ there exists $y\geq x$ such that $1\in I(y,B)$ or $0\in I(y,B)$. But since $0\notin I(x, B)$ for any $x\geq w$, it has to be that for all $x\geq w$ there exists $y\geq x$ such that $1\in I(y, B)$. Thus by supposition, for all $x\geq w$ there exists $y\geq x$ such that $1\in I(y, C)$. But this contradicts with $0\in I(w, C)$; so our supposition cannot hold. (ii) We want to show
\begin{equation*}
\forall{x\geq w}\exists{y\geq x}(1\in I(y, B\to C)\text{ or }0\in I(y, B\to C)).
\end{equation*}
Let $x\geq w$. Then by IH there is $y\geq x$ s.t. $1\in I(y,B)$ or $0\in I(y,B)$. Now again by IH there is $z\geq y$ s.t. $1\in I(z,C)$ or $0\in I(z,C)$ as well as $1\in I(z,B)$ or $0\in I(z,B)$ by monotonicity. Then if $1\in I(z,C)$ or $0\in I(z, B)$, we infer $1\in I(z, B\to C)$: the latter case follows from the IH of (i) for $B$. On the other hand, if $1\in I(z, B)$ and $0\in I(z, C)$, then from the former $0\notin I(u, B)$ for all $u\geq z$. Hence $0\in I(z, B\to C)$.\\
\textbf{For universal quantifier:}
(i) If $1\in I(w,\forall{x}A)$, then $1\in I(w,A(d))$ for all $d\in D(w)$. So by IH $0\notin I(w, A(d))$ for all $d\in D(w)$. Hence $0\notin I(w,\forall{x}A)$. (ii) Given $w\in W$, by frame condition there is a $x\geq w$ that is maximal. By IH and maximality, for all $d\in D(x)$, either $1\in I(x, A(d))$ or $0\in I(x, A(d))$. Thus $1\in I(x, A(d))$ for all $d\in D(x)$ or $0\in I(x, A(d))$ for some $d\in D(x)$. So $1\in I(x, \forall{x}A)$ or $0\in I(x, \forall{x}A)$. 
\end{proof}

\subsection{Proof system}

\begin{definition}
The logic {\bf QBDi3} is a system in $\mathcal{L}_{Q}$ defined by 
\eqref{Ax1}--\eqref{deMP} (except for \eqref{Ax3}), \eqref{MP},\eqref{Gen} as well as the following axioms. (We shall use $\Gamma\vdash_{i3} A$ for the derivability relation.)

\vspace{-3mm}

\begin{minipage}{.5\textwidth}
\begin{gather}
\forall{x}\neg\neg A\to\neg\neg\forall{x}A \tag{i1} \label{i7}\\
\Not A{\to}\neg A \tag{i2} \label{i9}
\end{gather}
\end{minipage}
\
\begin{minipage}{.4\textwidth}
\begin{gather}
\neg\neg(A\lor\Not A) \tag{i3} \label{i10}
\end{gather}
\end{minipage}
\smallskip
\end{definition}

\begin{remark}
If we change the language to $\mathcal{L}_{int}$ and axioms to non-$\Not$-related ones (except \eqref{deMbot}), then we obtain the intermediate logic {\bf MH} \cite{Gabbay1972}.  We shall use $\vdash_{mh}$ to denote the derivability in {\bf MH}.  
\end{remark}

\eqref{i10} is an axiom schema known as \emph{potential omniscience}, which was investigated in \cite{hasuo2003kripke} as one of the additional axiom to {\bf N3}. In comparison, we have the following remark on the status of \eqref{i10} in {\bf QBDi3}.  

\begin{remark}\label{rem.ax}
We note that \eqref{i10} is in fact redundant in {\bf QBDi3}: consider a subsystem of {\bf QBDi3} without \eqref{i10}, and take an instance $\Not\neg A\to\neg\neg A$ of \eqref{i9}. This is equivalent to $\neg(\neg\Not A\land \neg A)$, and so to the schema for \eqref{i10}. Alternatively, we may drop \eqref{i9} instead of \eqref{i10} in obtaining an equivalent system: an instance $\neg\neg(\neg A\lor\Not\neg A)$ of \eqref{i10} is equivalent to $\neg\neg\Not A\to \neg A$, so \eqref{i9} is derivable. In spite of these observations, We posit both of the axioms because it is more convenient for the proof of the completeness theorem.
\end{remark}

\begin{remark}
It is immediate from the above remark that the addition of $A\lor\Not A$ to {\bf BDi} results in the collapse of $\neg A$ and $\Not A$, as well as the classicalisation of the positive fragment. This can be contrasted with {\bf N4}, for which the same addition makes the positive fragment of the logic classical, but not $\Not$ \cite{BaCl04}. 
\end{remark}

\begin{remark}\label{rem.ax2}
It is shown in \cite{hasuo2003kripke} that the combination of \eqref{i9} and \eqref{i10} proves \eqref{i7}. To see this, note $\forall{x}\neg\neg A$ derives $\neg\exists{x}\neg A$ and so $\neg\exists{x}\Not A$ by \eqref{i9}. This is equivalent to $\neg\Not\forall{x}A$ and thus by \eqref{i10} $\neg\neg\forall{x}A$. Therefore \eqref{i7} is also redundant. We retain it again for convenience in the completeness proof.
\end{remark}

\subsection{Completeness}
In order to establish the completeness of {\bf QBDi3}, we first introduce the notion of \emph{reduction} \cite{gurevich1977}.

\begin{definition}
We define a \emph{reduction} $f:\mathsf{Form}_{Q}\to\mathsf{Form}_{Q}$ by the following clauses:
\vspace{-1mm}
\begin{IEEEeqnarray*}{rClrClrCl}
f(P) & = & P, \hspace{5mm}& f(\Not P) & = & \Not P, \hspace{5mm}& f(\Not\exists{x}A) & = &\forall{x} f(\Not A), \\
f(\bot) & = & \bot, \hspace{5mm} & f(\Not\bot) & = & \Not\bot, \hspace{5mm} & f(\Not(A\land B)) & = & f(\Not A)\lor f(\Not B),\\
f(A\circ B) & = & f(A)\circ f(B), \hspace{5mm}&  f(\Not\Not A) & = & f(A) , \hspace{5mm}&  f(\Not(A\lor B)) & = & f(\Not A)\land f(\Not B),\\
f(Q{x}A) & = & Q{x}f(A), \hspace{5mm}& f(\Not\forall{x}A) & = &\exists{x} f(\Not A) , \hspace{5mm}& f(\Not(A\to B)) & = & \neg f(\Not A)\land f(\Not B).
\end{IEEEeqnarray*}
where $\circ\in\{\land,\lor,\to\}$ and $Q\in\{\forall,\exists\}$. We then let $f(\Gamma)=\{f(B): B\in\Gamma\}$ for a set $\Gamma$ of formulas.
\end{definition}

Recall that a \emph{prime} formula is either atomic or $\bot$. The next proposition is then readily checkable.

\begin{proposition}
For all $A\in\mathsf{Form}_{Q}$, any $B$ in a subformula $\Not B$ of $f(A)$ is a prime formula. 
\end{proposition}

We shall call a formula \emph{reduced} if it is of the form $f(A)$. We shall often write $A[\Not P_{1},\ldots, \Not P_{n}]$ to denote the occurrences of subformulas of the form $\Not B$.
If all formulas in a proof are reduced, then we shall call it a \emph{reduced proof}, and use the notation $\vdash_{r}$. Then the proposition below is shown easily.

\begin{proposition}\label{prop.equiv}
For all $A\in\mathsf{Form}_{Q}$, $\vdash_{i3} A\leftrightarrow f(A)$.
\end{proposition}

\begin{proposition}\label{prop.red}
For all $\Gamma\cup\{A\}\subseteq\mathsf{Form}_{Q}$, if $\Gamma\vdash_{i3} A$ then $f(\Gamma)\vdash_{r}f(A)$. 
\end{proposition}

\begin{proof}
By induction on the length of a proof. For cases concerning \eqref{i9} and \eqref{i10}, we show
\begin{equation*}
    \vdash_{r} f(\Not A\to\neg A) \text{ and }\vdash_{r}f(\neg\neg (A\lor\Not A))
\end{equation*}
by simultaneous induction on the complexity of $A$. When $A$ is prime, $\Not A\to \neg A$ and $\neg\neg(A\lor\Not A)$ are already reduced. When $A\equiv\Not B$, 
$f(\Not A\to \neg A)=f(B)\to\neg f(\Not B)$, which is equivalent to $f(\Not B\to\neg B)$. Hence by IH there is a reduced proof. Similarly for $f(\neg\neg(A\lor\Not A))$.\\
\textbf{For conjunction:} When $A\equiv B\land C$, we have to show:\vspace{1mm}\\
\begin{minipage}{.5\textwidth}
\begin{enumerate}
    \item[1.] $\vdash_{r}f(\Not B)\lor f(\Not C)\to\neg(f(B)\land f(C))$,
\end{enumerate}  
\end{minipage}
\
\begin{minipage}{.5\textwidth}
\begin{enumerate}
    \item[2.] $\vdash_{r}\neg\neg((f(B)\land f(C))\lor f(\Not B)\lor f(\Not C))$.
\end{enumerate}
\end{minipage}\vspace{1mm}\\
By IH, there are reduced derivations for:\\
\begin{minipage}{.5\textwidth}
\begin{itemize}
    \item[1.] $f(\Not B)\to\neg f(B)$ and $f(\Not C)\to\neg f(C)$,
\end{itemize}  
\end{minipage}
\
\begin{minipage}{.5\textwidth}
\begin{itemize}
    \item[2.] $\neg\neg(f(B)\lor f(\Not B))$ and $\neg\neg(f(C)\lor f(\Not C))$.
\end{itemize}
\end{minipage}
For (1), the formula follows from
%\begin{equation*}
    $\vdash_{r}(\neg f(B){\lor}\neg f(C)){\to}\neg(f(B){\land} f(C))$.
%\end{equation*}
For (2), the formula follows from
%\begin{equation*}
    $\vdash_{r}((f(B){\lor}f(\Not B)){\land}(f(C){\lor} f(\Not C))){\to} ((f(B){\land} f(C)){\lor} f(\Not B){\lor} f(\Not C))$.
%\end{equation*}
The case for $\lor$ is similar.\\
\textbf{For implication:} When $A\equiv B\to C$, we have to show:\vspace{1mm}\\
\begin{minipage}{.5\textwidth}
\begin{enumerate}
    \item[1.] $\vdash_{r}(\neg f(\Not B){\land}f(\Not C))\to\neg(f(B){\to}f(C))$.
\end{enumerate}  
\end{minipage}
\
\begin{minipage}{.5\textwidth}
\begin{enumerate}
    \item[2.] $\vdash_{r}\neg\neg((f(B){\to}f(C))\lor (\neg f(\Not B){\land}f(\Not C)))$.
\end{enumerate}
\end{minipage}\vspace{1mm}\\
For (1), we shall show 
%\begin{equation*}
 $\vdash_{r}(f(\Not C)\land(f(B)\to f(C)))\to\neg\neg f(\Not B)$. 
%\end{equation*}
 First, by IH 
%\begin{equation*}
$\vdash_{r}(f(B)\to f(C))\to (f(\Not C)\to\neg f(B))$.   
%\end{equation*}
Then note 
%\begin{equation*}
$\vdash_{r}\neg\neg (f(B)\lor f(\Not B))\to (\neg f(B)\to\neg\neg f(\Not B))$.
%\end{equation*}
Hence by IH the desired formula follows.
For (2), we first note that $\neg((f(B)\to f(C))\lor (\neg f(\Not B)\land f(\Not C)))$ is equivalent to 
%\begin{equation*}
$\neg\neg f(B)\land\neg f(C)\land (\neg f(\Not B)\to \neg f(\Not C))$.
%\end{equation*}
(Recall $\neg(A\to B)\leftrightarrow(\neg\neg A\land \neg B)$ is an intuitionistic theorem.) Now by IH, $\vdash_{r}\neg\neg f(B)\to\neg f(\Not B)$; so $\neg f(C)\land\neg f(\Not C)$ follows from the above formula. But by IH we also have $\vdash_{r}\neg f(C)\to\neg\neg f(\Not C)$. Thus:
%\begin{equation*}
$\vdash_{r}\neg((f(B)\to f(C))\lor (\neg f(\Not B)\land f(\Not C)))\to(\neg f(\Not C)\land\neg\neg f(\Not C))$ 
%\end{equation*}
and so the desired formula follows by an intuitionistic inference.\\
\textbf{For universal quantifier:} When $A\equiv\forall{x}B$, we have to show:\vspace{1mm}\\
\begin{minipage}{.5\textwidth}
\begin{enumerate}
    \item[1.] $\vdash_{r}\exists{x}f(\Not B)\to\neg\forall{x}f(B)$.
\end{enumerate}  
\end{minipage}
\
\begin{minipage}{.5\textwidth}
\begin{enumerate}
    \item[2.] $\vdash_{r}\neg\neg(\forall{x}f(B)\lor \exists{x}f(\Not B))$.
\end{enumerate}
\end{minipage}\vspace{1mm}\\
For (1), from IH we can derive $\vdash_{r}\exists{x}f(\Not B)\to \exists{x}\neg f(B)$. Then use the fact that $\exists{x}\neg C\to\neg\forall{x}C$ is intuitionistically derivable.
For (2), by IH, \eqref{Gen} and \eqref{i7},
%\begin{equation*}
    $\vdash_{r}\neg\neg\forall{x}(\neg f(B)\to\exists{x}f(\Not B))$.
%\end{equation*}
Hence using \eqref{Ax12} and contraposing the inside, we obtain
%\begin{equation*}
     $\vdash_{r}\neg\neg(\neg\exists{x}f(\Not B)\to\neg\exists{x}\neg f(B))$.   
%\end{equation*}
Using the equivalence between $\neg\exists{x} C$ and $\forall{x}\neg C$ as well as \eqref{i7}, this implies
%\begin{equation*}
     $\vdash_{r}\neg\neg(\neg\exists{x}f(\Not B)\to\neg\neg\forall{x}f(B))$.   
%\end{equation*}
Since $C\to\neg\neg D$ is equivalent to $\neg\neg(C\to D)$,
%\begin{equation*}
     $\vdash_{r}\neg\neg(\neg\exists{x}f(\Not B)\to\forall{x}f(B))$.   
%\end{equation*}
Therefore %(recall Glivenko's theorem for propositional logic),
%\begin{equation*}
     $\vdash_{r}\neg\neg(\exists{x}f(\Not B)\lor\forall{x}f(B))$, using $(\neg C\to D)\to\neg\neg(C\lor D)$.   
%\end{equation*}
So the desired formula follows. The case for $\exists$ is similar.
\end{proof}

Given a set of reduced formulas $\Gamma$, we define a set of formula $E_{\Gamma}$ in $\mathcal{L}_{int}$ by:\vspace{-2mm}
\begin{IEEEeqnarray*}{rCr}
    E_{\Gamma} & := &\{\forall{\vec{x}}(P'{\to}\neg P): \Not P\text{ occurs in some }B\in\Gamma\}
      \cup\{\forall{\vec{x}}\neg\neg(P'{\lor} P): \Not P\text{ occurs in some }B\in\Gamma\}\vspace{-2mm}
\end{IEEEeqnarray*}

Given a reduced formula $A[\Not P_{1},\ldots \Not P_{n}]$, we define $A'$ to be the formula obtained by replacing the occurrences of $\Not P_{i}$ with $P'_{i}$.
We then define $\Gamma'=\{B':B\in\Gamma\}$ for a set $\Gamma$ of reduced formulas.

\begin{proposition}\label{prop.E}
Let $\Gamma\cup\{A\}\subseteq \textsf{Form}_{Q}$ be reduced. Then $\Gamma\vdash_{i3}A$ if and only if  $\Gamma',E_{\Gamma\cup\{A\}}\vdash_{mh} A'$.\vspace{-1mm}
\end{proposition}

\begin{proof}
For arguing left-to-right, by proposition \ref{prop.red} we can assume that the derivation of $A$ from $\Gamma$ to be reduced.\footnote{We may assume the subformulas of the form $\Not P$ in $\Gamma\cup\{A\}$ exhaust all formulas of the form in the derivation, for otherwise we can take $A\land (\Not P\to\Not P)$ instead. A similar remark applies to the right-to-left case.} then by induction the length of a proof, we can show that {\bf MH} can replicate the derivation of {\bf BDi3}. In particular, for \eqref{i9} and \eqref{i10}, the formulas negated by $\Not$ must be prime, and we have:
\begin{quote}
     $\forall{\vec{x}}(P'\to\neg P)\vdash_{mh} (\Not P\to\neg P)[\Not P/ P']$ and
     $\forall{\vec{x}}\neg\neg(P\lor P')\vdash_{mh}\neg\neg(P\lor\Not P)[\Not P/ P']$.
\end{quote}
Similarly for the case of $\bot$. For arguing right-to-left, by replacing atomic formulas of the form $P'$ by $\Not P$ in the proof of $\Gamma',E_{\Gamma\cup\{A\}}\vdash_{mh} A'$, we obtain a proof for $\Gamma\vdash_{i3}A$.
%\begin{equation*}
%    \vdash_{i3} E[p',\ldots,p'_{n}/\Not p_{1},\ldots, \Not p_{n}]\to A.
%\end{equation*}
%Hence $\Gamma\vdash_{i3}A$.
\end{proof}

We move on to the completeness theorem after stating one more lemma that is easily checkable.

\begin{lemma}\label{lem.red}
In a {\bf BDi3}-model and $A(\vec{d})\in\mathsf{Sent}_{{\bf D}}$ s.t. $\vec{d}\in D(w)$, the next equivalences hold.

\smallskip
\begin{minipage}{.5\textwidth}
\begin{itemize}
\item[(i)] $1\in I(w,A(\vec{d}))$ iff $1\in I(w, f(A(\vec{d})))$.
\end{itemize}
\end{minipage}
\
\begin{minipage}{.5\textwidth}
\begin{itemize}
\item[(ii)] $0\in I(w,A(\vec{d}))$ iff $1\in I(w, f(\Not A(\vec{d})))$.
\end{itemize}
\end{minipage}
\end{lemma}

%\begin{proof}
%By induction on the complexity of $A$.
%\end{proof}

\begin{theorem}[Soundness and completeness of {\bf QBDi3}]\label{thm.comp}
For all $\Gamma\cup\{A\}\in\mathsf{Sent}_{Q}$, $\Gamma\vdash_{i3} A$ iff $\Gamma\models_{i3} A$.
\end{theorem}
\vspace{-3mm}
\begin{proof}
The soundness follows by induction on the length of derivation (by substituting free variables with elements in the relevant domain). In particular, the cases for \eqref{i9}, \eqref{i10} follow from Proposition \ref{prop.model}.\\
\indent For completeness, we show by contraposition. Assume $\Gamma\nvdash_{i3}A$. Then by Proposition \ref{prop.equiv}, $f(\Gamma)\nvdash_{i3}f(A)$, and so $f(\Gamma)',E_{f(\Gamma\cup\{A\})}\nvdash_{mh} f(A)'$ by Proposition \ref{prop.E}. Hence by the strong completeness for {\bf MH} \cite{Ardeshir2014,Gabbay1972},  $f(\Gamma)',E_{f(\Gamma\cup\{A\})}\not\models_{mh}f(A)'$. Consequently, there is a model $\langle W,\leq, D, V\rangle$ of {\bf MH} such that for some $w\in W$, $1\in I(w, B)$ for all $B\in f(\Gamma)'\cup E_{f(\Gamma\cup\{A\})}$ but $1\notin I(w,f(A)')$ for some $x\in W$.\\
\indent Define a {\bf QBDi3}-model $\langle W,\leq,D, V_{2}\rangle$ such that for $\vec{d}\in D(w)$:
\begin{quote}
    %\item $1\in V_{2}(w, P(\vec{d}))$ iff $1\in V(w,P(\vec{d}))$.
     $\vec{d}\in V_{2}^{+}(w,P)$ iff $\vec{d}\in V^{+}(w,P)$,
    %\item $0\in V_{2}(w, P(\vec{d}))$ iff $1\in V(w, P'(\vec{d}))$.
    and $\vec{d}\in V_{2}^{-}(w,P)$ iff $\vec{d}\in V^{+}(w,P')$.
\end{quote}
We have to check that $\langle W,\leq,D,V_{2}\rangle$ is indeed a {\bf QBDi3}-model. 
If  $\vec{d}\in V_{2}^{+}(w,P)$ and  $\vec{d}\in V_{2}^{-}(w,P)$, then  $\vec{d}\in V^{+}(w,P)$ and $\vec{d}\in V^{+}(w,P')$.
%If $1\in V_{2}(w,P(\vec{d}))$ and $0\in V_{2}(w,P(\vec{d}))$, then $1\in V(w,P(\vec{d}))\cap V(w,P'(\vec{d}))$. 
But then $1\in I(w, B)$ for all $B$ in $\langle W,\leq,D,V\rangle$, a contradiction. Next, since $1\in I(w, \neg\neg (P(\vec{d})\lor P'(\vec{d})))$ for $\vec{d}\in D(w)$, for any $w\in W$: $\forall{x\geq w}\exists{y\geq x}(\vec{d}\in V^{+}(y, P)\cup V^{+}(y,P')))$.
%$\forall{x\geq w}\exists{y\geq x}(1\in V(y,P(\vec{d}))\text{ or }1\in V(y,P'(\vec{d})))$
Hence for any $\vec{d}\in D(w)$, we have $\forall{x\geq w}\exists{y\geq x}(\vec{d}\in V_{2}^{+}(y, P)\cup V_{2}^{-}(y,P))$ in $\langle W,\leq,D,V_{2}\rangle$.\\
%$\forall{x\geq w}\exists{y\geq x}(1\in V_{2}(y,P(\vec{d}))\text{ or }0\in V_{2}(y,P(\vec{d})))$.\\
\indent We shall now observe that 
%\begin{equation*}
    $1\in I(w, B')\text{ iff } 1\in I_{2}(w, B)$
%\end{equation*}
for any closed subformulas of $f(\Gamma\cup\{A\})$ with constants in $D(w)$. In particular, when $B\equiv \Not C$, $C\equiv P(\vec{d})$ for some $P$ which occurs in $E_{f(\Gamma\cup\{A\})}$. Then $1\in I(w,(\Not P(\vec{d}))')$ iff $0\in I_{2}(w, P(\vec{d}))$ iff $1\in I_{2}(w,\Not P(\vec{d}))$.\\
\indent It now follows that $1\in I_{2}(x, f(B))$ for all $f(B)\in f(\Gamma)$ but $1\notin I_{2}(x, f(A))$.  Therefore from Lemma \ref{lem.red}, we infer that $1\in I_{2}(x, B)$ for all $B\in \Gamma$ but $1\notin I_{2}(x, A)$. Hence $\Gamma\not\models_{i3}A$.
\end{proof}

%%%%%%%%%%%%%%%%%%%%%%%%%%%%%%%%%%%%
\subsection{Constructive properties}

Constructivity for {\bf BDi} has been observed in \cite{kamide2021modal} by establishing the disjunction and constructible falsity properties. These properties constitute an important difference from {\bf HYPE}, for which they fail, as Odintsov and Wansing \cite{odintsovWansing2021} observed through Drobyshevich's formula \cite{drobyshevich2015}. On the other hand, for {\bf MH}, the disjunction and existence properties have been established by Komori \cite{komori1983}. It is therefore of interest to check these properties for {\bf QBDi3}. Here, we adopt an approach via \emph{Aczel slash} \cite{aczel1968saturated}.\vspace{-0.5mm}

\begin{definition}
For $A\in\mathsf{Sent}_{Q}$. we define its \emph{slashes} $|^{+}A$ and $|^{-}A$ by the following clause.

\smallskip
\begin{minipage}{0.4\textwidth}
\begin{itemize}
\setlength{\parskip}{0cm}
\setlength{\itemsep}{0cm}
\item $|^{+}P(t_{1},\ldots t_{n})\text{ iff }\vdash_{i3}P(t_{1},\ldots t_{n}).$
\item $|^{-}P(t_{1},\ldots t_{n}) \text{ iff } \vdash_{i3}\Not P(t_{1},\ldots t_{n}).$
\item \hspace{-0.22em}$\not{|^{+}}\bot.$
\item $|^{-}\bot.$
\item $|^{+}\Not A \text{ iff } |^{-}A.$
\item $|^{-}\Not A \text{ iff } |^{+}A.$
\item $|^{+}A\land B \text{ iff } |^{+}A\text{ and }|^{+}B.$
\item $|^{-}A\land B \text{ iff } |^{-}A\text{ or }|^{-}B.$
\end{itemize}
\end{minipage}
\
\begin{minipage}{0.6\textwidth}
\begin{itemize}
\setlength{\parskip}{0cm}
\setlength{\itemsep}{0cm}
\item $|^{+}A\lor B \text{ iff } |^{+}A\text{ or }|^{+}B.$
\item $|^{-}A\lor B \text{ iff } |^{-}A\text{ and }|^{-}B.$
\item $|^{+}A\to B \text{ iff } \vdash_{i3}A\to B\text{ and }(|^{+}A\text{ implies }|^{+}B).$
\item $|^{-}A\to B \text{ iff } \vdash_{i3}\neg\Not A\text{ and }|^{-}B.$
\item $|^{+}\forall{x}A \text{ iff } \vdash_{i3}\forall{x}A\text{ and }(|^{+}A(c)\text{ for all }c\in\mathsf{Con}).$
\item $|^{-}\forall{x}A \text{ iff } |^{-}A(c)\text{ for some }c\in\mathsf{Con}.$
\item $|^{+}\exists{x}A \text{ iff } |^{+}A(c)\text{ for some }c\in\mathsf{Con}.$
\item $|^{-}\exists{x}A \text{ iff } \vdash_{i3}\Not\exists{x}A\text{ and }(|^{-}A(c)\text{ for all }c\in\mathsf{Con}).$
\end{itemize}
\end{minipage}
\end{definition}

\vspace{1mm}
We proceed to show a couple of lemmas. The first one has a handy consequence that $|^{+}\neg A$ iff $\vdash_{i3}\neg A$.

\begin{lemma}\label{lem.slash1}
Let $A\in\mathsf{Sent}_{Q}$. Then $|^{+}A$ implies $\vdash_{i3}A$.    
\end{lemma}

\begin{proof}
By induction on the complexity of $A$. When $A$ is strongly negated, we further divide into cases depending on the complexity of the negand. As an example, consider the case $A\equiv \Not(B\to C)$. Assume $|^{+}\Not(B\to C)$: then  $|^{-}(B\to C)$ and so $\vdash_{i3}\neg\Not B$ and $|^{-}C$. The latter implies $|^{+}\Not C$, which by IH implies $\vdash_{i3}\Not C$. Thus $\vdash_{i3}\Not(B\to C)$ follows from \eqref{deMI}.
\end{proof}

\indent Before stating the next lemma, we expand the (+ve) slash to $\text{Form}_{Q}$, by stipulating $|^{+}A$ if $|^{+}A'$ for any $A'$ obtained from $A$ by substituting its free variables by constants.

\begin{lemma}\label{lem.slash2}
Let $A\in\mathsf{Sent}_{Q}$. Then $\vdash_{i3}A$ implies $|^{+}A$.
\end{lemma}

\begin{proof}
 By induction on the length of proof, using the expanded notion of slash. Here we treat a couple of cases as examples. For cases of intuitionistic axioms and rules, see e.g. \cite[Theorem 3.5.9]{TrDa88i}. Moreover, in view of Remark \ref{rem.ax}, \ref{rem.ax2}, it suffices to consider a simpler axiomatisation of {\bf QBDi3} without \eqref{i7}, \eqref{i10}. \\
\indent For \eqref{deMI}, we need to show $|^{+}\Not(A\to B)\to(\neg\Not A\land \Not B)$ and $|^{+}(\neg\Not A\land \Not B)\to \Not(A\to B)$ for $A,B\in\mathsf{Sent}_{Q}$. Consider the former. By definition, it is equivalent to:
\begin{equation*}
\vdash_{i3}\Not(A\to B)\to(\neg\Not A\land \Not B)\text{ and } 
(|^{+}\Not(A\to B)\text{ implies }|^{+}\neg\Not A\land \Not B).
\end{equation*}
The former conjunct is one direction of \eqref{deMI}; the latter conjunct follows immediately from the handy consequence we noted above. The other direction similarly follows.\\
\indent For \eqref{i9}, we must show $|^{+}\Not A\to \neg A$ for $A\in\mathsf{Sent}_{Q}$. This follows since $|^{+}\Not A$ implies $\vdash_{i3}\Not A$ and thus $\vdash_{i3}\neg A$ by the previous lemma and \eqref{i9}: now use again the handy consequence to conclude $|^{+}\neg A$.
\end{proof}

We obtain disjunction, existence and constructible falsity property for {\bf QBDi3} as consequences.

\begin{theorem}\label{thm.cons}
Let $A,B\in\mathsf{Sent}_{Q}$. Then:\\

\begin{minipage}{.44\textwidth}
\begin{itemize}
\setlength{\parskip}{0cm}
\setlength{\itemsep}{0cm}
  \item[(i)] $\vdash_{i3} A{\lor}B$ implies $\vdash_{i3} A$ or $\vdash_{i3} B$.
  \item[(ii)] $\vdash_{i3}\exists{x}A$ then $\vdash_{i3}A(c)$ for some $c{\in}\mathsf{Con}$.
\end{itemize}
\end{minipage}
\begin{minipage}{.49\textwidth}
\begin{itemize}
\setlength{\parskip}{0cm}
\setlength{\itemsep}{0cm}
  \item[(iii)] $\vdash_{i3} \Not(A{\land}B)$ implies $\vdash_{i3} \Not A$ or $\vdash_{i3} \Not B$.
  \item[(iv)] $\vdash_{i3}\Not\forall{x}A$ then $\vdash_{i3}\Not A(c)$ for some $c{\in}\mathsf{Con}$.
\end{itemize}
\end{minipage}
\end{theorem}

\begin{proof}
(i) If $\vdash_{i3}A\lor B$, then by Lemma \ref{lem.slash2} $|^{+}A\lor B$, and so either $|^{+}A$ or $|^{+}B$. Thus either $\vdash_{i3} A$ or $\vdash_{i3}B$ by Lemma \ref{lem.slash2}. (ii) is shown analogously. (iii) and (iv) then follow form (i) and (ii), respectively.
\end{proof}

\begin{remark}
Despite Theorem \ref{thm.cons}, {\bf QBDi3} may be unacceptable to some constructivists, as the double negation shift contradicts principles of some schools of constructivism\footnote{For an analysis of the double negation shift and its variants in the mathematical setting, see e.g. \cite{fujiwarakohlenbach}.} \cite[Corollary 6.3.4.2, 6.6.4]{TrDa88i}.
\end{remark}

\section{Comparisons with systems related to {\bf BDi3}}
\subsection{Two-state case as a four-valued logic}
Let $\mathcal{L}$ be $\mathcal{L}_{Q}$ without quantifiers. Consider the extension of propositional {\bf BDi3} with an axiom schema:
\begin{gather}
    A\lor(A\to B)\lor\neg B. \label{AxG} \tag{AxG}
\end{gather}
For intuitionistic logic, the addition of \eqref{AxG} results in a system called {\bf G3}, which is sound and strongly complete with respect to the class of linear Kripke frames with $\leq 2$ elements: cf. \cite{Chagrov1997,ono2019proof,robles2014simple}. The semantics can be represented by the three-valued truth tables below.
\begin{center} 
{\scriptsize
$
\begin{tabular}{c|cccc}
$A \land B$ & $\mathbf{1}$ & $\mathbf{i}$  & $\mathbf{0}$ \\ \hline
$\mathbf{1}$ & $\mathbf{1}$ & $\mathbf{i}$ & $\mathbf{0}$\\
$\mathbf{i}$ & $\mathbf{i}$ & $\mathbf{i}$ & $\mathbf{0}$\\
$\mathbf{0}$ & $\mathbf{0}$ & $\mathbf{0}$  & $\mathbf{0}$\\
\end{tabular}
\quad
\begin{tabular}{c|cccc}
$A \lor B$ & $\mathbf{1}$ & $\mathbf{i}$  & $\mathbf{0}$ \\ \hline
$\mathbf{1}$ & $\mathbf{1}$ & $\mathbf{1}$ & $\mathbf{1}$\\
$\mathbf{i}$ & $\mathbf{1}$ & $\mathbf{i}$ & $\mathbf{i}$\\
$\mathbf{0}$ & $\mathbf{1}$ & $\mathbf{i}$ & $\mathbf{0}$
\end{tabular}
\quad
\begin{tabular}{c|cccc}
$A {\to} B$ & $\mathbf{1}$ & $\mathbf{i}$  & $\mathbf{0}$ \\ \hline
$\mathbf{1}$ & $\mathbf{1}$ & $\mathbf{i}$ & $\mathbf{0}$\\
$\mathbf{i}$ & $\mathbf{1}$ & $\mathbf{1}$ & $\mathbf{0}$\\
$\mathbf{0}$ & $\mathbf{1}$ & $\mathbf{1}$ & $\mathbf{1}$
\end{tabular}
\quad
\begin{tabular}{c|c}
${\neg} A$ & \\ \hline
$\mathbf{1}$ & $\mathbf{0}$\\
$\mathbf{i}$ & $\mathbf{0}$\\
$\mathbf{0}$ & $\mathbf{1}$
\end{tabular}
$
}
\end{center}
We shall use $\vdash_{i3g3}$  for the consequence in {\bf BDi3}+\eqref{AxG}, and $\models_{i3g3}$ for the semantical consequence of the class of linear propositional {\bf BDi3}-frames with $\leq 2$ elements. Then using the strong completeness of {\bf G3}, we can show the completeness theorem by arguing analogously to the previous subsection.

\begin{theorem}
For all $\Gamma\cup\{A\}\subseteq\mathsf{Form}$, $\Gamma\vdash_{i3g3} A$ iff $\Gamma\models_{i3g3} A$.
\end{theorem}

Given this correspondence, it is of interest to ask what kind of truth tables can characterize this extension. We claim that the following 4-valued truth tables are adequate ($\bot$ has the constant value $\mathbf{0}$).
\begin{center} 
{\scriptsize
$
\begin{tabular}{c|cccc}
$A \land B$ & $\mathbf{1}$ & $\mathbf{i}$  & $\mathbf{j}$ & $\mathbf{0}$ \\ \hline
$\mathbf{1}$ & $\mathbf{1}$ & $\mathbf{i}$ & $\mathbf{j}$ & $\mathbf{0}$\\
$\mathbf{i}$ & $\mathbf{i}$ & $\mathbf{i}$ & $\mathbf{j}$ & $\mathbf{0}$\\
$\mathbf{j}$ & $\mathbf{j}$ & $\mathbf{j}$ & $\mathbf{j}$ & $\mathbf{0}$\\
$\mathbf{0}$ & $\mathbf{0}$ & $\mathbf{0}$ & $\mathbf{0}$ & $\mathbf{0}$\\
\end{tabular}
\quad
\begin{tabular}{c|cccc}
$A \lor B$ & $\mathbf{1}$ & $\mathbf{i}$  & $\mathbf{j}$ & $\mathbf{0}$ \\ \hline
$\mathbf{1}$ & $\mathbf{1}$ & $\mathbf{1}$ & $\mathbf{1}$ & $\mathbf{1}$\\
$\mathbf{i}$ & $\mathbf{1}$ & $\mathbf{i}$ & $\mathbf{i}$ & $\mathbf{i}$\\
$\mathbf{j}$ & $\mathbf{1}$ & $\mathbf{i}$ & $\mathbf{j}$ & $\mathbf{j}$\\
$\mathbf{0}$ & $\mathbf{1}$ & $\mathbf{i}$ & $\mathbf{j}$ & $\mathbf{0}$\\
\end{tabular}
\quad
\begin{tabular}{c|cccc}
$A {\to} B$ & $\mathbf{1}$ & $\mathbf{i}$  & $\mathbf{j}$ & $\mathbf{0}$ \\ \hline
$\mathbf{1}$ & $\mathbf{1}$ & $\mathbf{i}$ & $\mathbf{j}$ & $\mathbf{0}$\\
$\mathbf{i}$ & $\mathbf{1}$ & $\mathbf{1}$ & $\mathbf{j}$ & $\mathbf{0}$\\
$\mathbf{j}$ & $\mathbf{1}$ & $\mathbf{1}$ & $\mathbf{1}$ & $\mathbf{1}$\\
$\mathbf{0}$ & $\mathbf{1}$ & $\mathbf{1}$ & $\mathbf{1}$ & $\mathbf{1}$\\
\end{tabular}
\quad
\begin{tabular}{c|c}
${\neg} A$ & \\ \hline
$\mathbf{1}$ & $\mathbf{0}$\\
$\mathbf{i}$ & $\mathbf{0}$\\
$\mathbf{j}$ & $\mathbf{1}$\\
$\mathbf{0}$ & $\mathbf{1}$\\
\end{tabular}
\quad
\begin{tabular}{c|c}
${\Not} A$ & \\ \hline
$\mathbf{1}$ & $\mathbf{0}$\\
$\mathbf{i}$ & $\mathbf{j}$\\
$\mathbf{j}$ & $\mathbf{i}$\\
$\mathbf{0}$ & $\mathbf{1}$\\
\end{tabular}
$
}
\end{center}

Let $V_{4}: \mathsf{Prop}\longrightarrow\{{\bf 1},{\bf i},{\bf j},{\bf 0}\}$ be a four-valued assignment and $I_{4}$ be the interpretation extending it according to the tables. We write $\Gamma\models_{4} A$ if $I_{4}(B)={\bf 1}$ for all $B\in \Gamma$ implies $I_{4}(A)={\bf 1}$ for all interpretations.

\begin{theorem}
For all $\Gamma\cup\{A\}\subseteq\mathsf{Form}$, if $\Gamma\models_{i3g3} A$ then $\Gamma\models_{4} A$.
\end{theorem}

\begin{proof}
For the left-to-right direction, let $V_{4}$ be an assignment s.t. $I_{4}(B)={\bf 1}$ for all $B\in\Gamma$. We define a linear {\bf BDi3}-model with $2$ elements $\langle \{x,y\},\{(x,x),(x,y),(y,y)\}, V\rangle$ by:

\vspace{-1mm}

{\small
\begin{equation*}
V(x,p):=
\begin{cases}
\{ 1\} &\text{ if }V_{4}(p)= {\bf 1}.\\
\{ 0\} &\text{ if }V_{4}(p)= {\bf 0}.\\
\emptyset &\text{ otherwise. }
\end{cases}
\hspace{3mm}
V(y,p):=
\begin{cases}
\{ 1\} &\text{ if }V_{4}(p)= {\bf 1}\text{ or }{\bf i}.\\
\{ 0\} &\text{ otherwise. }\\
\end{cases}
\end{equation*}
}
We can then show that $V$ is monotone and potentially omniscient, and for all $A\in\mathsf{Form}$:

\smallskip
\begin{minipage}{.5\textwidth}
\begin{itemize}
\setlength{\parskip}{0cm}
\setlength{\itemsep}{0cm}
    \item $I(x,A)=\{1\}\Longleftrightarrow I_{4}(A)={\bf 1}$.
    \item $I(x,A)=\{0\}\Longleftrightarrow I_{4}(A)={\bf 0}$.
    \item $I(x,A)=\emptyset\Longleftrightarrow I_{4}(A)={\bf i}\text{ or }{\bf j}$.
\end{itemize}
\end{minipage}
\
\begin{minipage}{.5\textwidth}
\begin{itemize}
\setlength{\parskip}{0cm}
\setlength{\itemsep}{0cm}
    \item $I(y,A)=\{1\}\Longleftrightarrow I_{4}(A)={\bf 1}\text{ or }{\bf i}$.
    \item $I(y,A)=\{0\}\Longleftrightarrow I_{4}(A)={\bf j}\text{ or }{\bf 0}$. 
\end{itemize}
\end{minipage}

\smallskip 
\noindent Now by assumption, $1\in I(x, B)$ for all $B\in\Gamma$ and so $1\in I(x,A)$; hence $I_{4}(A)={\bf 1}$. Thus  $\Gamma\models_{4} A$.
\end{proof}

\begin{theorem}
For all $\Gamma\cup\{A\}\subseteq\mathsf{Form}$, if $\Gamma\models_{4} A$ then $\Gamma\models_{i3g3} A$.
\end{theorem}

\begin{proof}
Let $\langle W,\leq, V\rangle$ be a linear {\bf BDi3}-model with $\leq 2$ elements such that $1\in I(w, B)$ for all $B\in\Gamma$. As the case when $|W|=1$ is immediate, we turn our attention to the case when $|W|=2$. Let $W=\{x,y\}$, $\leq =\{(x,x),(x,y),(y,y)\}$ and $w=x$. We define an assignment $V_{4}$ by the following clauses.

{\small
\begin{equation*}
    V_{4}(p)= \begin{cases}
              {\bf 1} & \text{ if } V(x,p)=\{1\}.\\
              {\bf i} & \text{ if } V(x,p)=\emptyset\text{ and }V(y,p)=\{1\}.\\
              {\bf j} & \text{ if } V(x,p)=\emptyset\text{ and }V(y,p)=\{0\}.\\
              {\bf 0} & \text{ if } V(x,p)=\{0\}.\\
    \end{cases}
\end{equation*}
}
This can be checked to generalise to all $A\in\mathsf{Form}$.
Now by assumption, $I_{4}(B)={\bf 1}$ for all $B\in\Gamma$ and thus $I_{4}(A)={\bf 1}$. Hence $I(x, A)=\{1\}$. Therefore $\models_{g3i3} A$.
\end{proof}

Therefore we conclude that {\bf BDi3}+\eqref{AxG} is sound and complete with respect to the above tables:

\begin{corollary}
For all $\Gamma\cup\{A\}\subseteq\mathsf{Form}$, $\Gamma\vdash_{i3g3} A$ iff $\Gamma\models_{4} A$.
\end{corollary}

\subsection{Some subsystems of {\bf BDi3}}
Here we make some observations regarding the predicate expansions of other systems related to {\bf QBDi3}.\\
\indent Firstly, we consider the predicate version {\bf QBDi} of the system {\bf BDi}. A major difference of {\bf QBDi} from {\bf QBDi3} is that there is no need to posit the double negation shift axiom.

\begin{definition}
A {\bf QBDi}-model is a quadruple $\langle W,\leq, D, V\rangle$  defined like that of {\bf QBDi3}, except that:
\begin{itemize}
    \item The condition about the existence of maximal elements is dropped.\vspace{-1ex}
    \item The condition $V^{+}(w,P^n)\cap V^{-}(w,P^n)=\emptyset$ and the assumption of potential omniscience are dropped.
\end{itemize}
We shall use $\models_{i}$ in denoting the semantic consequence.
\end{definition}

\begin{definition}
The logic {\bf QBDi} is a system in $\mathcal{L}_{Q}$ defined by removing \eqref{i7},\eqref{i9},\eqref{i10} from the axiomatisation of {\bf QBDi3}. We shall use $\vdash_{i}$ to denote the derivability in {\bf QBDi}.
\end{definition}

\begin{theorem}
For all $\Gamma\cup\{A\}\subseteq\mathsf{Sent}_{Q}$, $\Gamma\vdash_{i} A$ iff $\Gamma\models_{i} A$.
\end{theorem}

\begin{proof}
The argument is analogous to Theorem \ref{thm.comp}. We do not need an analogue of Proposition \ref{prop.model}, and the proof of the analogue of Proposition \ref{prop.red} is much simplified. For the analogue of Proposition \ref{prop.E} and elsewhere, we do not need to appeal to $E_{\Gamma\cup\{A\}}$. In the proof of the theorem itself, we appeal to the strong completeness of intuitionistic logic, rather than of {\bf MH}.
\end{proof}

%\begin{proof}
Constructive properties of {\bf QBDi} can be observed as well,  by arguing analogously to Theorem  \ref{thm.cons}.
%\end{proof}
Next, we consider the predicate expansions {\bf QDN3} and {\bf QDN4} of the systems {\bf DN3} and {\bf DN4} \cite{Niki2023}. {\bf QDN4} is defined from {\bf QBDi} by replacing \eqref{deMI}  with
%\begin{equation*}
    $\Not(A\to B)\leftrightarrow(\neg\neg A\land\Not B)$. 
%\end{equation*}
A Kripke model for {\bf QDN4} is obtained from  that of {\bf QBDi} by changing the clauses for $0\in I(w, A{\to}B)$ to:
\begin{itemize}
\item $0\in I(w, A{\to} B)$ iff for all $x\geq w$ there is $y\geq x(1\in I(y, A))$) and $0\in I(w,B)$.
\end{itemize}
{\bf QDN3} and its models are defined by imposing \eqref{i9} and the condition $V^{+}(w,P^n)\cap V^{-}(w,P^n)=\emptyset$.

Let us use subscripts $_{d3}$ and $_{d4}$ for the syntactic and semantic consequences in these systems. Then we obtain the following completeness theorems (cf. also \cite{Niki2023} for the propositional case.)

\begin{theorem}
Let $k\in\{3,4\}$. For all $\Gamma\cup\{A\}\subseteq\mathsf{Sent}_{Q}$, $\Gamma\vdash_{dk} A$ iff $\Gamma\models_{dk} A$.
\end{theorem}
\begin{proof}
For {\bf QDN4}, the argument is the same as the case for {\bf QBDi}. The only major difference is that we have to use the clause $f(\Not(A\to B))=\neg\neg f(A)\land f(\Not B)$ for reduction. For {\bf QDN3}, the outline is almost identical to the case of {\bf QBDi3}. Aside from the difference in reduction, and using the completeness of intuitionistic logic rather than of {\bf MH}, we take $E_{\Gamma}$ to be $\{\forall{\vec{x}}(P'\to\neg P): \exists{B\in\Gamma}(\Not P\text{ occurs in }B)\}$.
\end{proof}

\begin{remark}
A motivation for {\bf DN3} and {\bf DN4} is to brings strong and intuitionistic negation closer: $\Not (A\to B)\rightarrow A$ holds in {\bf N4}, but its analogue does not hold w.r.t. $\neg$. This may appear too demanding for a refutation of implication, and is thus avoided in the systems of \cite{Niki2023}. This approach is also more thoroughly pursued in \emph{quasi-nelson algebras} \cite{rivieccio2019}: notice a similarity with the clause for $\to$ in \emph{nucleus-based quasi-Nelson twist-algebra} \cite{rivieccio2022}, where $\Box$ is a \emph{nucleus} (a generalisation of double negation):

\vspace{-1mm}
\begin{itemize}
\item $\langle a_{1}, a_{2}\rangle\to\langle b_{1},b_{2}\rangle=\langle a_{1}\to b_{1},\Box a_{1}\land b_{2}\rangle$.
\end{itemize}
\end{remark}

Constructive properties of {\bf QDN3} and {\bf QDN4} can be checked again analogously to Theorem  \ref{thm.cons}, by changing the clause for $|^{-}A\to B$ by $\vdash\neg\neg A\text{ and }|^{-}B$.
%\end{proof}
Next, we observe that {\bf QBDi3} and {\bf DN3} are related in an essential way; indeed, the difference is exactly the potential omniscience axiom.

\begin{proposition}
${\bf QBDi3}={\bf QDN3}+\eqref{i10}$.
\end{proposition}
\begin{proof}
It suffices to show that $\neg\Not A\leftrightarrow\neg\neg A$ in each system, for then the two conditions for negated implications become inter-derivable. For {\bf QBDi3}, it follows from \eqref{i9} using $\neg\Not A\leftrightarrow\Not\neg A$. For ${\bf QDN3}+\eqref{i10}$, one direction follows from \eqref{i9} and the other direction is equivalent to \eqref{i10}.   
\end{proof}

\begin{remark}
This also means that another advantage of {\bf DN3} over {\bf N3} claimed in \cite{Niki2023}, namely that contraposition is available in a limited form $(\neg A\to B)\to(\Not B\to\Not\neg A)$, also holds for {\bf QBDi3}.   
\end{remark}

On the other hand, {\bf QDN4} is not a subsystem of {\bf QBDi}; that would imply $\vdash_{i}\Not\neg A\leftrightarrow\neg\neg A$ and thus $\vdash_{i}\neg\Not A\to\neg\neg A$, i.e. \eqref{i10} that separates {\bf QBDi} from {\bf QBDi3}.

\begin{remark}
In \cite{Niki2023}, we observed another extension of {\bf DN4} by the axiom schema $A\lor\Not A$. At the propositional level, this already derives the weak law of excluded middle $\neg\neg A\lor \neg A$. If we consider a predicate expansion of this logic, then for the semantics 
 %to define the valuation to be s.t. $V:W\times \mathsf{Atom_{\mathbf{D}}}\longrightarrow\{\{0\},\{1\},\{0,1\}\}$. Then in order 
to validate $\forall{x}A\lor\Not\forall{x}A$ we seem to require that a model has a constant domain.\footnote{This situation is similar to the case for the predicate extension {\bf QC3} of a connexive logic {\bf C3}. \cite{olkhovikov2023completeness,omori2020extension}} This suggests the adoption of the constant domain axiom $\forall{x}(A(x)\lor C)\to (\forall{x}A(x)\lor C)$ in the expansion. On the other hand, the combination of the weak excluded middle and the constant domain axiom is known to cause Kripke incompleteness in intermediate logics \cite{ghilardi1989presheaf,shehtman1990semantics}.  So an adequate treatment of the predicate system for this extension is expected to need more sophistications.

 \end{remark}

\subsection{A connexive variant?}

One of the most well-known variant of {\bf N4} is the logic {\bf C} introduced by Wansing \cite{Wansing2005}. This is obtained by replacing the conjunction in the {\bf N4} condition $\Not(A\to B)\leftrightarrow(A\land\Not B)$ by implication. As a result of this change, {\bf C} validates \emph{Aristotle's theses} $\Not(A\to\Not A)$, $\Not(\Not A\to A)$ and \emph{Boethius' theses} $(A\to B)\to\Not(A\to\Not B)$ and $(A\to\Not B)\to\Not(A\to B)$ characteristic to connexive logic \cite{SEPconnexive}.

We can also test what happens if a similar change is made to {\bf BDi}. In this case, \eqref{deMI} becomes $\Not(A{\to} B){\to}(\neg\Not A{\to}\Not B)$ and otherwise the axiomatisation is kept intact. Then the theses become equivalent to $\neg\Not A{\to} A$, $\neg A{\to}\Not A$ (for Aristotle's theses) and $(A{\to} B){\to}(\neg\Not A{\to} B)$, $(A{\to}\Not B){\to}(\neg\Not A{\to} \Not B)$ (for Boethius' theses). So the resulting  system is not connexive, but only humbly connexive (cf. \cite{Kapsner2019humble}).\\
\indent Another characteristic of {\bf C} is that it is non-trivial but \emph{negation inconsistent}, i.e. it validates a formula and its (strong) negation. That this would also be negation inconsistent in our variant of {\bf BDi} is evident as $\Not\bot$ is one of the axioms. We also find a witness for negation inconsistency even in the absence of this axiom: e.g. both $(p\land\Not\neg\Not p)\to\Not \neg\Not p$ and $\Not((p\land\Not\neg\Not p)\to\Not \neg\Not p)$ turn out to be derivable. This system (and its extension with the variants of the connexive theses) remains \emph{non-trivial}; this is checkable with the classical truth tables which in addition assigns every formula of the form $\Not A$ the value ${\bf 1}$.

\section{Concluding remarks}

Our main motivation was to connect {\bf BD+} and its intuitionistic counterpart {\bf BDi} (in the first-order setting) with neighbouring systems. We firstly focused on establishing the picture of {\bf BDi} and {\bf HYPE} as sibling systems, through the formulation of star semantics for {\bf QBD+}. Our suggestion there was to understand the two systems as results of constuctivising {\bf BD+} along different (American/Australian) semantical contours. One question that remains, connecting back to the example of {\bf S5} in the introduction, is whether there are other siblings for the two systems: i.e. a logic with the intuitionistic positive part, whose extension by Peirce's law coincides with {\bf BD+}. Another venue would be to compare {\bf BDi} and {\bf HYPE} in more details, by e.g. introducing star semantics for {\bf BDi} following ones for {\bf N4} by Routley \cite{Routley1974semantical}.\\
\indent The second focus in this article was to compare {\bf QBD+} from a more Nelsonian viewpoint. For this purpose an explosive system {\bf QBDi3} was introduced. We observed a remarkable feature of this system that the falsity condition for implication now settles the status of potential omniscience and double negation shift. Since the motivations for these principles are by themselves not too clear, the falsity condition can provide another route to analyse their desirability. A further understanding of the falsity condition may be facilitated by comparison with the \emph{strong implication} $A\Rightarrow B:=(A\to B)\land(\Not B\to\Not A)$ in {\bf BDi} and {\bf BDi3} (also for {\bf DN4} and {\bf DN3}), following the approach for {\bf N3}/{\bf N4} in \cite{spinks2008constructive,spinks2008constructive2,spinks2018paraconsistent}.

\bibliographystyle{eptcs}
\bibliography{Wansing-style,FDE-modal-logics}
\end{document}